\documentclass[10pt, twoside, journal,twocolumn]{IEEEtran} %final version two column
\IEEEoverridecommandlockouts
\usepackage{algpseudocode}
\usepackage{algorithm}
\usepackage{color}
\usepackage{bm}
\usepackage{amsfonts}
\usepackage{amssymb}
\usepackage{amscd}
\usepackage{graphics}
\usepackage[mathcal]{eucal}
\usepackage[dvips]{graphicx}
\usepackage{epsfig}
\usepackage{subfigure}
\usepackage{caption}
\usepackage[cmex10]{amsmath}
\usepackage{array}
\usepackage{eqparbox}
\usepackage{flushend}
\usepackage{fancyhdr}
\usepackage{titling}
\usepackage{enumerate}
\usepackage{stackrel}
\usepackage{pdflscape}
\usepackage{longtable}
\usepackage{lineno}

\newcommand{\ie}{{\em i.e.}}
\newcommand{\etc}{{\em etc}}

\newcommand{\iid}{i.i.d.}
\newcommand{\apriori}{{\em a priori}}

\newcommand{\secref}[1]{Section~\ref{#1}}
\newcommand{\figref}[1]{Fig.~\ref{#1}}

\newcommand{\thrmref}[1]{Theorem~\ref{#1}}

\newcommand{\propref}[1]{Proposition~\ref{#1}}

\newcommand{\appref}[1]{Appendix~\ref{#1}}

\newtheorem{thrm}{\textbf{Theorem}}
\newtheorem{prop}{\textbf{Proposition}}

\newtheorem{defn}{\textbf{Definition}}

\newcommand{\abs}[1]{\left\vert#1\right\vert}
\newcommand{\norm}[1]{\Vert#1\Vert}

\makeatletter
\def\blfootnote{\xdef\@thefnmark{}\@footnotetext}
\makeatother

\newenvironment{proof}[1][Proof]{\begin{trivlist}
\item[\hskip \labelsep {\bfseries #1}]}{\end{trivlist}}

\newcommand{\qed}{\nobreak \ifvmode \relax \else
      \ifdim\lastskip<1.5em \hskip-\lastskip
      \hskip1.5em plus0em minus0.5em \fi \nobreak
      \vrule height0.75em width0.5em depth0.25em\fi}

\def\BibTeX{{\rm B\kern-.05em{\sc i\kern-.025em b}\kern-.08em
    t\kern-.1667em\lower.7ex\hbox{E}\kern-.125emX}}

\usepackage{newlfont}
\hypersetup{pageanchor=false}
\begin{document}
\title{A Learning-Based Approach to Caching in Heterogenous Small Cell Networks}
\author{\IEEEauthorblockN{B. N. Bharath, K. G. Nagananda and H. Vincent Poor, \emph{Fellow, IEEE}}\thanks{B. N. Bharath is with People's Education Society Institute of Technology, Bangalore South Campus, INDIA, E-mail: \texttt{bharathbn@pes.edu}. K. G. Nagananda is with People's Education Society University, INDIA, E-mail: \texttt{kgnagananda@pes.edu}. H. Vincent Poor is with Princeton University, New Jersey, USA, E-mail: \texttt{poor@princeton.edu}}}

\date{}
\setlength{\droptitle}{-.25in}
\maketitle

%\vspace{-0.8in}
\begin{abstract}%\vspace{-0.15in}
A heterogenous network with base stations (BSs), small base stations (SBSs) and users distributed according to independent Poisson point processes is considered. SBS nodes are assumed to possess high storage capacity and to form a distributed caching network. Popular files are stored in local caches of SBSs, so that a user can download the desired files from one of the SBSs in its vicinity. The offloading-loss is captured via a cost function that depends on the random caching strategy proposed here. The popularity profile of cached content is unknown and  estimated using instantaneous demands from users within a specified time interval. An estimate of the cost function is obtained from which an optimal random caching strategy is devised. The training time to achieve an $\epsilon>0$ difference between the achieved and optimal costs is finite provided the user density is greater than a predefined threshold, and scales as $N^2$, where $N$ is the support of the popularity profile. A transfer learning-based approach to improve this estimate is proposed. The training time is reduced when the popularity profile is modeled using a parametric family of distributions; the delay is independent of $N$ and scales linearly with the dimension of the distribution parameter.
\end{abstract}
%\vspace{-0.15in}
\begin{IEEEkeywords} %\vspace{-0.15in}
Caching; small cell networks; popularity profile; transfer learning.
\end{IEEEkeywords}

\section{Introduction} \label{sec:intorduction} \vspace{-0.05in}
The advent of multimedia-capable devices at economical costs has triggered the growth of wireless data traffic at an unprecedented rate. This trend is likely to continue, requiring wireless service providers to reevaluate design strategies for the next generation wireless infrastructure \cite{Furuskar2015}. A promising approach to address this problem is to deploy small cells that can offload a significant amount of data from a macro base station (BS) \cite{Chou2014}. Doing so, it is expected to lead to cost-effective integration of the existing WiFi and cellular technologies with improved performance of peak data traffic steering policies \cite{Bennis2013}. However, a potential shortcoming of the small cell infrastructure is that, during peak traffic hours, the backhaul link-capacity requirement to support data traffic is enormously high \cite{Kim2014}. Also, the cost incurred in deploying a high capacity backbone network for small cells can be quite high. Therefore, small cell-based solutions alone will not suffice to efficiently solve the quality of service requirements associated with peak traffic demands.

A noteworthy development in this direction is to improve the accessibility of data content to users by storing the most popular data files in the \emph{local caches} (intermediate servers such as gateways, routers, {\etc}.) of small cell BSs, with the objective of reducing the peak traffic rates. This is commonly referred to as ``caching'' and has attracted significant attention \cite{Lin2003} - \nocite{Bastug2014}\nocite{Sadjadpour2014}\cite{Niesen2012}. In the next subsection we mention a few references, which although by no means exhaustive, fairly indicate the scope and trend of research on caching.
%\vspace{-0.175in}

\subsection{Literature review on caching}\label{subsec:related_work}%\vspace{-0.1in}
Caching has received considerable attention in the wireless communications literature. In \cite{Hu2003}, a two-layer hierarchical strategy termed New Snoop was proposed to cache the unacknowledged packets from mobiles and BSs to significantly enhance TCP performance. In \cite{Wang2014}, a technique based on the concept of content-centric networking was devised for caching in 5G networks, while in \cite{Golrezaei2014} caching of video files was proposed by exploiting the redundancy of user requests and storage capacity of mobile devices with {\apriori} knowledge of the locations of devices. In \cite{Lee2014}, the effects of cache size and cached-data popularity on a data access scheme were studied to mitigate the traffic load over the wireless channel. In \cite{Poularakis2014}, inner and outer bounds were proposed for the joint routing and caching problem in small cell networks, while in \cite{Fang2014} in-network caching was proposed for an information-centric networking architecture for faster content distribution in an energy-efficient manner. In-network caching was employed in \cite{Asaeda2015} for content-centric networks using a tool called ``contrace'' for monitoring and operating the network. The tradeoff between the performance gain of coded caching and delivery delay in video streaming was characterized in \cite{Pedarsani2014}. A polynomial-time heuristic solution was proposed in \cite{Guan2014} to address the NP-hard optimization problem of maximizing the caching utility of mobile users.

Caching has also made advances in device-to-device (D2D) communications. In \cite{Pyattaev2015}, a practical method was devised for data caching and content distribution in D2D networks to enhance assisted communications between proximate nodes. In \cite{Ji2013}, the outage-throughput tradeoff was characterized for D2D nodes, which obtained the desired file from nodes which had that file in its cache. In \cite{Golrezaei2012}, the conflict between collaboration-distance and interference was identified among D2D nodes to maximize frequency reuse by exploiting distributed storage of cached content. In \cite{Ji2013a}, coded caching was shown to achieve multicast gain in a D2D network, where users had access to linear combinations of packets from cached files. In \cite{Ji2014}, the throughput scaling laws of random caching, where users with pre-cached information made arbitrary requests for cached files, were studied. New caching mechanisms developed by modeling the network as independent Poisson point processes (PPPs) with full knowledge of the popularity profile can be found in \cite{Bacstuug2015a} - \nocite{Yang2015} \nocite{Chae2015}\cite{Tamoor-ul-Hassan2015}, while the most recent results on caching in D2D networks and video content delivery are reported in \cite{Zhang2015} and \cite{Li2015}.

Caching has been addressed from an information-theoretic viewpoint as well. In \cite{Maddah-Ali2014}, it was shown that when cached-content demand is uniformly distributed, joint optimization of caching and coded multicast delivery significantly improves the gains; this setup was extended to the case of nonuniform distributions on demand and to a decentralized setting in \cite{Niesen2014a} and \cite{Maddah-Ali2014a}, respectively. In \cite{Karamchandani2014}, coded caching was achieved for content delivery networks with two layers of caches.

%\vspace{-0.15in}
\subsection{Main contributions of this paper}\label{subsec:main_contributions}%\vspace{-0.05in}
In the aforementioned references, the popularity profile of data files was assumed to be known perfectly. In practice, such an assumption cannot be reasonably justified; this was clearly highlighted in \cite{Blasco2014} - \nocite{Blasco2014a}\nocite{Sengupta2014}\nocite{Bastug2014a}\cite{Bacstuug2015}, where various learning-based approaches were employed to estimate the popularity profile. On the other hand, estimation procedures result in computational overhead especially in data-intensive realtime multimedia applications. Therefore, given the increasing demand for improving the quality of service for the end users, establishing the theoretical underpinnings of learning-based caching strategies is a topical research problem, and is the main subject of this paper.

In this work, we relax the assumption of {\apriori} knowledge of the popularity profile to devise a caching strategy. We consider a heterogenous network where the users, BS and small base stations (SBSs) are assumed to be distributed according to PPPs. Each SBS is assumed to employ a random caching strategy with no caching at the user terminal (see \cite{Ji2013}). A protocol model for communications is proposed using which a cost that captures backhaul link overhead that depends on the popularity profile is derived. Assuming a Poisson request model, a centralized approach is presented in which the BS computes an estimate of the popularity profile based on the requests observed during the time interval $[0,\tau]$; this estimate is then used in the cost function to optimize the caching probability. Thus, the actual cost incurred differs from the optimal cost, and this difference depends on the number of samples used to estimate the popularity profile. Further, the number of samples collected at the BS depends on the density of the Poisson arrival process and the training time during which the samples are collected. A lower bound on this training time is derived that guarantees a cost that is within $\epsilon > 0$ of the optimal cost. The results are improved using a transfer learning (TL)-based approach wherein samples from other domains, such as those obtained from a social network, are used to improve the estimation accuracy; the minimum number of source domain samples required to achieve better performance is derived. Finally, we model the popularity profile using a parametric family of distributions (specifically, the Zipf distribution \cite{Llorca2013}) to analyze the benefits offered.

The following are the main findings of our study:
\begin{enumerate}[(i)]
\item The training time $\tau$ is finite, provided the user density is greater than a predefined threshold.
\item $\tau$ scales as $N^2 \log N$, where $N$ is the total number of cached data files in the system.
\item Employing the TL-based approach, a finite training time can be achieved for \emph{all} user densities. In this case, the training time is a function of the ``distance'' between the probability distribution of the files requested and that of the source domain samples (the notion of distance will be made precise in the proof of \thrmref{thm:time_complexity_centralized_TL}).
\item When the popularity profile is modeled using a parametric family of distributions, the bound on the training time is independent of $N$, and scales only linearly with the dimension of the distribution parameter leading to a significant improvement in the performance compared to its nonparametric counterpart.
\end{enumerate}

The problem of periodic caching without the knowledge of the popularity profile, but with access to the demand history, was addressed in \cite{Blasco2014} and \cite{Blasco2014a}; however, the model and objective function considered in our work are different from those presented therein. Learning-based approaches to estimate the popularity profile for devising caching mechanisms have also been reported in \cite{Sengupta2014} - \nocite{Bastug2014a}\cite{Bacstuug2015}; while caching in femtocell networks without prior knowledge of the popularity distribution was considered in \cite{Golrezaei2013}, where it was shown that distributed caching was NP-hard and approximation algorithms were proposed for video content delivery. We would like to emphasize that the central focus of this paper is not on deriving new caching mechanisms. Our main contribution is the theoretical analysis of the implications of learning the popularity profile on the training time to achieve an offloading loss which is $\epsilon > 0$ close to the optimal policy. To the best of our knowledge, this is the first instance where an analytical treatment of training time and its relation to the probability distribution function of source domain samples has been reported in the literature on caching. Some preliminary aspects of this work can be found in \cite{Bharath2015}.

In \secref{sec:sys_model}, we present the system model followed by the main problem addressed in the paper. The two methods for estimating the popularity profile and its corresponding training time analysis are developed in \secref{sec:estimate_popularity}. The training time analysis when the popularity profile is modeled as a parametric family of distributions is presented in \secref{sec:param_pac_bound}. Numerical results are reported in \secref{sec:sims}. Concluding remarks are provided in \secref{sec:conclude}. The proofs of the theorems are relegated to appendices.

\section{System Model and Problem Statement} \label{sec:sys_model}
In this section, we present the system model followed by the main problem addressed in the paper. The notation used in the rest of the paper is as follows: $\Phi_u$, $\Phi_s$  and $\Phi_b$ ($\lambda_u$, $\lambda_s$ and $\lambda_b$) denote the points (densities) corresponding to the user, SBSs and BS, respectively; $k_x$ denotes the number of requests in $[0,\tau]$ by the user at $x$; $X_x^{(l)}$ denotes the $l^\text{th}$ request of the user $x$; $\lambda_r$ is the average number of requests per unit time. A heterogenous cellular network is considered where the set $\Phi_u \subseteq \mathbb{R}^2$ of users, the set $\Phi_b \subseteq \mathbb{R}^2$ of BSs, and the set $\Phi_s \subseteq \mathbb{R}^2$ of SBSs are distributed according to independent PPPs with density $\lambda_u$, $\lambda_b$ and $\lambda_s$, respectively, in the two-dimensional space \cite{Baccelli1997}. Each user independently requests a data-file of size $B$ bits from the set $\mathcal{F} \triangleq \{f_1,f_2,\ldots,f_N\}$; the popularity of data files is specified by the distribution $\mathcal{P} \triangleq \{p_1,\ldots,p_N\}$, where $\sum_{i=1}^N p_i = 1$ and is assumed to be stationary across time. In a typical heterogenous cellular network, the BS fetches a file using its backhaul link to serve a user. During peak data traffic hours, this results in an information-bottleneck both at the BS as well as in its backhaul link. To alleviate this problem, caching the most popular files (either at the user nodes or at SBSs) is proposed. The requested file will be served directly by one of the neighboring SBSs depending on the availability of the file in its local cache. The performance of caching depends on the density of SBS nodes, cache size, users' request rate, and the caching strategy. It is assumed that the SBS can cache up to $M$ files, each of length $B$ bits. Each SBS in $\Phi_s$ caches its content in an independent and identically distributed ({\iid}) fashion by generating $M$ indices distributed according to $\Pi \triangleq \{\pi_i: f_i \in \mathcal{F},~i=1,2,\ldots,N\}$, $\sum_{i=1}^N \pi_i = 1$ (see \cite{Ji2013}). One way of generating this is to roll an $N$ sided die $M$ times in an {\iid} fashion, where the outcomes correspond to the index of the file to be cached. Although this approach is suboptimal, it is mathematically tractable and the corresponding time complexity serves as a lower bound, albeit pessimistic, for optimal strategies.

We now present a simple communications protocol to determine the set of neighboring SBS nodes for any user in $\Phi_u$. Essentially, we let each SBS at location $y \in \Phi_s$ communicate with a user at location $x \in \Phi_u$ if $\norm{y - x} < \gamma$, ($\gamma > 0$); this condition determines the communication radius. In this protocol, we have ignored the interference constraint. The set of neighbors of the user at location $x$ is denoted
\begin{equation} \label{eq:Nxdefinition}
\mathcal{N}_x \triangleq \{y \in \Phi_s: \norm{y - x} < \gamma\}.
\end{equation}
%\vspace{-0.5in}

\subsection{The main problem addressed in this paper}\label{subsec:problem_statement}
The user located at $x \in \Phi_u$ requests a data-file from the set $\mathcal{F}$, with the popularity profile chosen from the probability distribution function $\mathcal{P}$. The requested file will be served directly by a neighboring SBS at location $y \in \Phi_s$ depending on the availability of the file in its local cache, and following the protocol described in the previous paragraph. The problem of caching involves minimizing the time overhead incurred due to the unavailability of the requested file. Without loss of generality and for ease of analysis, we focus on the performance of a typical user located at the origin, denoted by $o \in \Phi_u$. The unavailability of the requested file from a user located at $o$ is given by
\begin{eqnarray}
\mathcal{T}(\Pi,\mathcal{P})  \triangleq \frac{B}{R_0}\mathbb{E} \sum_{i=1}^N\left[\mathbf{1}\{f_i \notin \mathcal{N}_o\} \right] \mathbf{1}\{f_i \text{ requested} \},\hspace{-2mm}
%&=&\hspace{-2mm} \frac{B}{R_0}\mathbb{E} \left[\mathbf{1}\{f_i \notin \mathcal{N}_u\} \}\right],
\label{eq:metric}
\end{eqnarray}
where $\mathcal{N}_o$ is as defined in \eqref{eq:Nxdefinition}, $R_0$ is the rate supported by the BS to the user, and $\frac{B}{R_0}$ is the time overhead incurred in transmitting the file from the BS to the user. Further, we use $f_i \notin \mathcal{N}_o$ to denote the event that the file $f_i$ is not stored in any of the SBSs in $\mathcal{N}_o$. The expectation is with respect to $\Phi_u$, $\Phi_s$ and $\mathcal{P}$. The indicator function $\mathbf{1}\{A\}$ is equal to one if the event $A$ occurs, and zero otherwise. We refer to $\mathcal{T}(\Pi,\mathcal{P})$ as the ``offloading loss'', which we seek to minimize:
\begin{eqnarray} \label{eq:opt_problem}
\min_{\Pi \succeq 0} ~~ \mathcal{T}(\Pi,\mathcal{P}) \\ \nonumber
\text{       subject to } \sum_{i=1}^N \pi_i = 1,
\end{eqnarray}
where $\pi_i \geq 0$, for $i=1,\ldots,N$. To solve the optimization problem \eqref{eq:opt_problem}, we need an analytical expression for $\mathcal{T}(\Pi,\mathcal{P})$ which is provided in the following theorem.

\begin{thrm}  \label{thm_mean_throughput}
For the caching strategy proposed in this paper, the average offloading loss is given by
\begin{eqnarray} \label{eq:mean_througput_expression}
\mathcal{T}(\Pi,\mathcal{P})  = \frac{B}{R_0} \left[ \sum_{i=1}^N \exp\{-\lambda_s \pi \gamma^2\left[1 - (1-\pi_i)^M\right]\} p_i\right] .
\end{eqnarray}
%where $\mathcal{G}(\lambda_u, q, \gamma):= 1 - \exp\{-\lambda_u \pi \gamma^2 (1-q^T)\}$.
\end{thrm}
\begin{proof}
See Appendix \ref{app:througput_derivation}.
\end{proof}

We note that, solving the optimization problem posed in \eqref{eq:opt_problem} is not the main focus of this paper. We assume that there exists a method to solve the problem posed in \eqref{eq:opt_problem}, and instead focus on analyzing the training time required to obtain a good estimate of the popularity profile that results in an offloading loss that is within $\epsilon$ of the optimal offloading loss. Interestingly, although the problem in \eqref{eq:opt_problem} is non-convex, since it is separable a bound on the duality gap can be obtained with respect to the solution derived using the Karush-Kuhn-Tucker conditions.

In practice, the popularity profile $\mathcal{P}$ is generally unknown and has to be estimated. Denoting the estimated popularity profile by $\hat{\mathcal{P}} \triangleq \{\hat{p}_1,\ldots,\hat{p}_N\}$, and the corresponding offloading loss by ${\mathcal{T}}(\Pi,\hat{\mathcal{P}})$, \eqref{eq:opt_problem} becomes
\begin{eqnarray}
\min_{\Pi \succeq 0} ~~ {\mathcal{T}}(\Pi,\hat{\mathcal{P}}) \label{eq:opt_problem_emperical} \\ \nonumber
\text{subject to} \sum_{i=1}^N \pi_i = 1,
\end{eqnarray}
with $\pi_i \geq 0$, for $i=1,\ldots,N$. Naturally, the solution to \eqref{eq:opt_problem_emperical} differs from that of the original problem \eqref{eq:opt_problem}. Let $\Pi^*$ and $\hat{\Pi}^*$ denote the optimal solutions to the problems in \eqref{eq:opt_problem} and \eqref{eq:opt_problem_emperical}, respectively, and let the throughput achieved using $\hat{\Pi}^*$ be denoted $\hat{\mathcal{T}}^* \triangleq \mathcal{T}(\hat{\Pi}^*,{\hat{\mathcal{P}}})$. \textbf{The central theme of this paper is the analysis of the offloading loss difference, {\ie}, $\hat{\mathcal{T}}^* - \mathcal{T}^*$, where $\mathcal{T}^* \triangleq \mathcal{T}(\Pi^*, \mathcal{P})$ is the minimum offloading loss incurred with perfect knowledge of the popularity profile $\mathcal{P}$.} Theorems \ref{thm:time_complexity_centralized} - \ref{thm:pac_param_pop_distr} are devoted to this analysis.

\section{Estimating the popularity profile} \label{sec:estimate_popularity}
In this section, we present two methods for estimating the popularity profile and provide the corresponding training time analyses. The efficiency of the estimate $\hat{\mathcal{P}}$ of the popularity profile depends on the number of available data samples, which in turn is related to the number of requests made by the users. We first obtain an expression for the estimate of the popularity profile. We then study, in \secref{subsec:waitingtime_lowerbound}, the minimum training time in obtaining the samples to achieve a desired estimation accuracy $\epsilon > 0$. Finally, in \secref{subsec:transfer_learning}, we employ the TL-based approach to improve the bound on the training time. We begin with the definition of the request model.

\defn (\emph{Request Model}) Each user requests a file $f \in \mathcal{F}$ at a random time $t \in [0,\infty]$ following an independent Poisson arrival process with density $\lambda_r>0$.

For notational convenience, the same density is assumed across all the users. The following centralized scheme is used where the BS collects the requests from all the users in its coverage area in a time interval $[0,\tau]$ to estimate the popularity profile of the requested files: Let the number of users in the coverage area of BS $z \in \Phi_b$ of radius $R>0$ be $n_R$, which is distributed according to a PPP with density $\lambda_u$. Let the number of requests made by the user at the location $x \in \{\Phi_u \bigcap \mathbb{B}(0,R)\}$ in the time interval $[0,\tau]$ be $k_x$, where $\mathbb{B}(0,R)$ is a two-dimensional ball of radius $R$ centered at $0$. We assume that requests across the users are known at the BS. The requests from the user $x$ is denoted $\mathcal{X}_x \triangleq \{X_x^{(1)},\ldots,X_x^{(k_x)}\}$, where $X_x^{(l)} \in \{1,\ldots,N\}$ denotes the indices of the files in $\mathcal{F}$, $l=0,\ldots,k_x$. After receiving $\mathcal{X}_x$, $x \in \{\Phi_u \bigcap \mathbb{B}(0,R)\}$, in the time interval $[0,\tau]$, the BS computes an estimate of the popularity profile as follows:
\begin{eqnarray}
\hat{p}_i = \frac{1}{\sum_{x \in \{\mathbb{B}(0,R) \bigcap \Phi_u\}}  k_x} {\sum_{x \in \left\{\mathbb{B}(0,R) \bigcap \Phi_u\right\} }\sum_{l=0}^{k_x} \mathbf{1}\{X_x^{(l)} = i\}}, \!\!\!\!
\label{eq:estimation_popularity}
\end{eqnarray}
$i=1,\ldots,N$. Given the number $n_R$ of users in the coverage area of the BS, the sum $\sum_{x \in \{\mathbb{B}(0,R) \bigcap \Phi_u\} } k_x$ is a PPP with density $n_R \lambda_r$. Also, $\mathbb{E} \left\{\hat{p}_i  |  \abs{\{\Phi_u \bigcap \mathbb{B}(0,R)\}} = n_R\right\} = p_i$, which leads us to conclude that $\hat{p}_i$ is an unbiased estimator. The estimated popularity profile $\hat{p}_i$ given by \eqref{eq:estimation_popularity} is shared with every SBS in the coverage area of the BS, and is then used in \eqref{eq:opt_problem_emperical} to find the optimal caching probability.

The proposed estimator can be improved by using samples from other related domains, for example, a social network. The term ``target domain'' is used when samples are obtained only from users in the coverage area of the BS. In the next subsection we derive the minimum training time $\tau$, corresponding to the estimator in \eqref{eq:estimation_popularity}, required to achieve the desired estimation accuracy $\epsilon > 0$.

\subsection{A lower bound on the training time $\tau$} \label{subsec:waitingtime_lowerbound}
\begin{thrm} \label{thm:lower_bound_nonparam_ntl}
For any $\epsilon>0$, with a probability of at least $1-\delta$, a throughput of $\hat{\mathcal{T}}^* \leq \mathcal{T}^* + \epsilon$ can be achieved using the estimate in \eqref{eq:estimation_popularity} provided
\begin{eqnarray} \label{eq:sourceonly_waitingtime}
\tau \geq \left\{ \begin{array}{cc} \left\{\frac{1}{\lambda_r g^*} \log \left(\frac{1}{1- \frac{1}{\lambda_u \pi R^2} \log \frac{2 N}{\delta}}\right)\right\}^+  &\text{if } \lambda_u > \mathcal{L},\\
\infty  &\text{otherwise},
\end{array}
\right.
\end{eqnarray}
where $\{x\}^+ \triangleq \max\{x,0\}$, $g^* \triangleq (1- \exp\{-2 \bar{\epsilon}^2\})$, $\mathcal{L}\triangleq \frac{1}{\pi R^2} \log \left(\frac{2 N}{\delta}\right)$ and
\begin{eqnarray}
\bar{\epsilon} \triangleq \frac{R_0 \epsilon}{2 B \sup_{\Pi} \sum_{i=1}^N g(\pi_i)},
\end{eqnarray}
with $g(\pi_i)\triangleq \exp\{-\lambda_s \pi \gamma^2 \left[1 - (1-\pi_i)^M \right]\}$.
\label{thm:time_complexity_centralized}
\end{thrm}
\begin{proof}
See Appendix \ref{app:time_complexity_centralized}.
\end{proof}

To achieve a finite training time that results in an estimation accuracy $\epsilon>0$, the user density $\lambda_u$ has to be greater than a threshold. Further insights into \eqref{eq:sourceonly_waitingtime} are obtained by making the following approximation: $1-x\leq e^{-x}$ for all $x\geq0$. This is combined with $\sup_{\Pi: \Pi \succeq 0, \mathbf{1}^T\Pi = 1} \sum_{i=1}^N g(\pi_i) \leq N$
%$\exp\{-\lambda_u \pi \gamma^2 (1-1/N)^M\}$
yielding the following lower bound on the training time $\tau$:
\begin{eqnarray}
\tau \geq  \frac{2 B^2}{\pi R^2 \lambda_u \lambda_r R_0^2 \epsilon^2} N^2 \log\left(\frac{2N}{\delta}\right).
\label{eq:source_domain_only_insight_waiting_time}
\end{eqnarray}
The lower bound \eqref{eq:source_domain_only_insight_waiting_time} enables us to make the following observations:
\begin{enumerate}[(i)]
\item The training time $\tau$  to achieve an $\epsilon$-offloading loss difference scales as $N^2$,
\item $\tau$ is inversely proportional to ($\lambda_u$, $\lambda_r$), and
\item as the coverage radius increases, the delay decreases as $1/R^2$, and
\item as the data-file size $B$ increases, the training time scales as $B^2$.
\end{enumerate}

The bound in \eqref{eq:source_domain_only_insight_waiting_time} is a lower bound on the training time per request per user, since the offloading loss is derived for a given request per user. There are on an average $\lambda_r$ requests per unit time per user. Thus, to obtain the training time per user, the offloading loss has to be multiplied by $\lambda_r$. This amounts to replacing $\epsilon$ by $\epsilon/\lambda_r$. Therefore, \eqref{eq:source_domain_only_insight_waiting_time} becomes
\begin{eqnarray} \label{eq:source_domain_only_insight_waiting_time_per_user}
\tau \geq  \frac{2 B^2 \lambda_r}{\pi R^2 \lambda_u R_0^2 \epsilon^2} N^2 \log\left(\frac{2N}{\delta}\right).
\end{eqnarray}
It is seen that the training time scales linearly with $\lambda_r$. Although the training time per user per request tends to zero as $\lambda_r \rightarrow \infty$, the training time per user tends to $\infty$. This is because the number of requests per unit time approaches $\infty$, and thus, a small fraction of errors results in an infinite difference in offloading loss leading to an infinite training time. With the increasing demand to provide higher quality of service for the end user, the question of whether it is possible to improve ({\ie} decrease) the training time $\tau$ to achieve the desired estimation accuracy $\epsilon$ deserves attention. In the next subsection we show that the lower bound on the training time can indeed be improved by employing a TL-based approach.

%\vspace{-0.2in}
\subsection{Transfer learning to improve the training time} \label{subsec:transfer_learning}
In practice, the minimum training time required to achieve an estimation accuracy $\epsilon > 0$ can be expected to be very large. An approach to overcome this drawback is to utilize the knowledge obtained from users' interactions with a social community (termed the ``source domain''). Specifically, by cleverly combining samples from the source domain and users' request pattern (target domain), one can potentially reduce the training time. In fact, the estimation accuracy is indicative of the dependence between the source and target domains. These techniques are commonly referred to as TL-based approaches, and have implications on the training time to achieve a given estimation accuracy. TL-based approaches were also employed in \cite{Bastug2014a} and \cite{Bacstuug2015} to negotiate over-fitting problems in estimating the content popularity profile matrix. However, unlike in \cite{Bastug2014a} and \cite{Bacstuug2015}, in this paper we are interested in deriving the minimum training time to achieve a desired performance accuracy. Furthermore, the model we consider is quite different from those considered in \cite{Bastug2014a} and \cite{Bacstuug2015}.

The TL-based approach considered here comprises two sources, namely, the source domain and target domain, from which the samples are acquired. An estimate of the popularity profile is obtained in a stepwise manner as follows:
\begin{enumerate}[(i)]
\item Using target domain samples, the following parameter is computed at the BS:
    \begin{eqnarray}
    \hat{S}_i^{(tar)} \triangleq \!\!\!\!  {\sum_{x \in \mathbb{B}(0,R) \bigcap \Phi_u } \sum_{l=0}^{k_x} \mathbf{1}\{X_x^{(l)} = i\}},
    i=1,\ldots,N.
    \label{eq:sum_estimate_target}
    \end{eqnarray}
Recall that $k_x$ is the number of requests made by the user at the location $x$. The corresponding $l^{\text{th}}$ request by the user at the location $x$ in the time interval $[0,\tau]$ is denoted $X_x^{(l)}$, $l=1,2,\ldots,k_x$.
\item The source domain samples $\mathcal{X}^s \triangleq \{X_1^{s},\ldots,X_m^{s}\}$ are drawn {\iid} from a distribution $\mathcal{Q}$, where $X_l^s = i~(i=1,\ldots,N$) denotes that the user corresponding to the $l^{\text{th}}$ sample has requested the file $f_i$. The nature of the distribution will be made precise in \propref{prop:TL_cache}. Using this, the BS computes
    \begin{eqnarray}
    \hat{S}_{i}^{s} \triangleq \sum_{k=1}^m \mathbf{1}\{X_k^{s} = i\},~i=1,2,\ldots,N.
    \label{eq:source_domain_sum}
    \end{eqnarray}

\item The BS uses \eqref{eq:sum_estimate_target} and \eqref{eq:source_domain_sum} to compute an estimate of $\hat{p}_i^{(tl)}$ (the superscript $tl$ indicates transfer learning) given by
    \begin{eqnarray} \label{eq:tl_flexible_estimate}
    \hat{p}_i^{(tl)}= \frac{\hat{S}_i^{(tar)} +  \hat{S}_{i}^{s}}{{\sum_{x \in \{\mathbb{B}(0,R) \bigcap \Phi_u\}}  k_x} + m}.
    \label{eq:estimation_TL}
    \end{eqnarray}
\end{enumerate}
Using the estimate given by \eqref{eq:estimation_TL}, a lower bound on the training time is obtained as stated in the next theorem.

\begin{thrm} \label{thm:lower_bound_nonparam_tl}
Let  $g(\pi_i) \triangleq \exp\{-\lambda_s \pi \gamma^2 \left[1 - (1-\pi_i)^M \right]\}$. Then, for any accuracy
\begin{eqnarray} \label{eq:epsilon_error_bound}
\epsilon > \frac{2 B  \sup_\Pi \left\{\sum_{i=1}^N g(\pi_i)\right\} } {R_0} \norm{\mathcal{P} - \mathcal{Q}}_{\infty},
\end{eqnarray}
with a probability of at least $1-\delta$, a throughput of $\hat{\mathcal{T}}^* \leq \mathcal{T}^* + \epsilon$ can be achieved using the estimate in \eqref{eq:estimation_TL} provided the training time $\tau$ satisfies the following condition:
\begin{eqnarray}
\tau \geq \left\{ \begin{array}{cc} \left\{\frac{1}{\lambda_r (1-e^{-2 \epsilon_{pq}^2})} \log \left(\frac{1}{1 - \Lambda} \right)\right\}^+, &\text{if } \lambda_u > \rho,\\
\infty,  &\text{otherwise},
\end{array}
\right.
\end{eqnarray}
where $\rho \triangleq \frac{1}{\pi R^2} \left(\log \frac{2N}{\delta} - 2 \epsilon_{pq}^2 m\right)$, $\epsilon_{pq}\triangleq \bar \epsilon - \norm{\mathcal{P} - \mathcal{Q}}_{\infty}$, $\Lambda \triangleq \frac{1}{\lambda_u \pi R^2} \left(\log \frac{2N}{\delta} - 2 \epsilon_{pq}^2 m\right)$, and $\bar \epsilon \triangleq \frac{R_0 {\epsilon}}{2 B \sup_\Pi \left\{\sum_{i=1}^N g(\pi_i)\right\} }$.
\label{thm:time_complexity_centralized_TL}
\end{thrm}
\begin{proof}
See Appendix \ref{app:TL_time_complexity}.
\end{proof}

From \thrmref{thm:time_complexity_centralized_TL}, we see that under suitable conditions the TL-based approach performs better than the source domain sample-based agnostic approach. The following inferences are drawn:
\begin{enumerate}[(1)]
\item The minimum user density to achieve a finite delay is reduced by a positive offset $2 \epsilon_{pq}^2 m$. In fact, for $m > \frac{\log{\left(\frac{2N}{\delta}\right)}}{2(\bar \epsilon - \norm{\mathcal{P}-\mathcal{Q}}_\infty)^2}$, a finite delay can be achieved for all user densities which provides a significant advantage.

\item The finite delay achieved is smaller compared to the source domain sample-based agnostic approach for large enough numbers of source samples, and the distributions are ``close.'' This is made precise in the following proposition, and a detailed discussion is provided in Section \ref{sec:sims}.
    \begin{prop}
    For any $\epsilon>0$ and $\delta \in [0,1]$, the TL-based approach performs better than the source sample-based agnostic approach provided the number $m$ of source samples satisfies
%    \begin{eqnarray}
     $ m \geq \frac{1}{2 \epsilon_{pq}^2} \left[\log\left(\frac{2N}{\delta}\right) - F\right]^+$,
  %  \end{eqnarray}
 and the distributions satisfy the following condition:
 \begin{eqnarray}
\norm{\mathcal{P}-\mathcal{Q}}_{\infty} < \frac{\epsilon R_0}{2 B \lambda_u \pi \gamma^2 N},
 \label{eq:diff_distribution}
\end{eqnarray}
 where $F \triangleq \lambda_u \pi R^2 \left(1 - \exp\left\{\frac{1 - e^{-2 {\epsilon}_{pq}^2}} {1 - e^{-2 \bar{\epsilon}^2}} \log \left(1 - \mathcal{L}\right) \right\}\right)$ and $\mathcal{L} \triangleq \frac{1}{\lambda_u \pi R^2} \log\left(\frac{2N}{\delta}\right)$.
\label{prop:TL_cache}
\end{prop}
In fact, \eqref{eq:diff_distribution} provides the guiding principle to decide if the samples drawn from the distribution $\mathcal{Q}$ should be used to estimate the distribution $\mathcal{P}$. In general, the distance between the distributions has to be estimated from the available samples (relative to the distribution on $\mathcal{P}$).
\end{enumerate}

An estimate of the popularity profile can also be obtained by linearly combining its estimates obtained from the source domain and target domain samples. In particular, we have
\begin{equation}
\hat{p}_i = \alpha \hat{p}_i^{(s)} + (1 - \alpha) \hat{p}_i^{(t)},
\label{eq:TL_convex_comb_estimate}
\end{equation}
where $\hat{p}_i^{(s)}$ and $\hat{p}_i^{(t)}$ are the estimates of the popularity profile obtained from the source domain samples and the target domain samples, respectively. The estimates are given by
\begin{eqnarray}
\hat{p}_i^{(t)} &=& \frac{\hat{S}_i^{(tar)}}{{\sum_{u \in \{\mathbb{B}(0,R) \bigcap \Phi_u\}}  k_u}}, \\
\hat{p}_i^{(s)} &=& \frac{\hat{S}_{i}^{s}}{m}.
\end{eqnarray}

Note that, in this case the coefficients are independent of the realization of the network. For the estimate proposed in \eqref{eq:TL_convex_comb_estimate}, we have the following result:
\begin{thrm} \label{thm:conv_comb}
For any accuracy
\begin{eqnarray} \label{eq:epsilon_error_bound}
\epsilon > \frac{2 B  \sup_\Pi \left\{\sum_{i=1}^N g(\pi_i)\right\} } {R_0} \norm{\mathcal{P} - \mathcal{Q}}_{\infty},
\end{eqnarray}
with a probability of at least $1-\delta$, a throughput of $\hat{\mathcal{T}}^* \leq \mathcal{T}^* + \epsilon$ can be achieved using the estimate in \eqref{eq:TL_convex_comb_estimate} provided the training time $\tau$ satisfies the condition specified by \eqref{eq:thrm4cond} at the top of the next page,
\begin{figure*}[!t]
\begin{eqnarray}
\tau \geq \left\{ \begin{array}{cc}  \frac{1}{\lambda_r g_t^*} \log \left[\frac{1}{1 - \frac{1}{\lambda_u \pi R^2 } \left(\log\left(\frac{2N}{\delta}\right) + \log\left\{\frac{1}{1 - \left(\frac{2N}{\delta}\right)\exp\{-2{\bar \omega}^2 m\}}\right\} \right)}\right], &\text{if } \lambda_u > \rho_\text{thres},\\
\infty,  &\text{otherwise},
\end{array}
\right. \label{eq:thrm4cond}
\end{eqnarray}
\hrulefill
\end{figure*}
where $$\rho_\text{thresh} \triangleq \frac{1}{\pi R^2} \left(\log\left(\frac{2N}{\delta}\right) + \log\left\{\frac{1}{1 - \left(\frac{2N}{\delta}\right)\exp\{-2{\bar \omega}^2 m\}}\right\} \right),$$
$\bar \epsilon \triangleq \frac{R_0 {\epsilon}}{2 B \sup_\Pi \left\{\sum_{i=1}^N g(\pi_i)\right\} }$, $g_t^* \triangleq \left( 1 - \exp\left\{-{2 \eta^2 }\right\}\right)$, $g(\pi_i) \triangleq \exp\{-\lambda_s \pi \gamma^2 \left[1 - (1-\pi_i)^M \right]\}$ and  $\omega \triangleq \frac{\bar{\epsilon} - (1 - \alpha) \eta}{\alpha} > 0$. This is valid for all $0 < \alpha < \min \left\{\frac{\bar \epsilon}{G},1\right\}$ and $0 \leq \eta < \frac{\bar \epsilon - \alpha  {G}}{1-\alpha}$, where $G \triangleq \norm{\mathcal{P}-\mathcal{Q}}_{\infty} + \sqrt{\frac{1}{2m} \log\frac{2N}{\delta}}$.
\label{thm:time_complexity_centralized_TL_convex}
\end{thrm}
\begin{proof}
See \appref{app:estimation_TL_convex}.
\end{proof}

\section{Parametrized Family of Popularity Profile}\label{sec:param_pac_bound}
In the previous sections, no structure was imposed on the popularity profile. In practice, the popularity profile is modeled using a parametric family of distributions such as the Zipf distribution \cite{Llorca2013}, which, with a one-dimensional parameter $\Theta \in \mathbb{R}$, is specified by ${p}_{\Theta,i} = \frac{{1/i^\Theta}}{\sum_{j=1}^N {1/j^\Theta}}, i = 1,2,\ldots,N$. To obtain an estimate of the Zipf distribution it suffices to find the parameter $\Theta$; estimating a single parameter requires fewer samples which can potentially reduce the training time. We now derive bounds on the training time when the popularity profile belongs to a parametric family of distributions. We begin with the following assumption: \\
\textbf{Assumption 1:} Let the family of parametrized popularity distributions be defined by $\mathcal{P} = \{\mathcal{P}_\Theta: \Theta \subseteq [a, b]^d, a < b\}$. Further, for all $\Theta \subseteq [a, b]^d$, $\mathcal{P}_\Theta$ satisfies $\sum_{i=1}^N \norm{{\partial}_\Theta {p}_{\Theta,i}}_2 <  C$, where $C<\infty$ is independent of $N$, and $\partial_\Theta {p}_{\Theta,i} \in \mathbb{R}^d$ denotes the sub-differential of ${p}_{\Theta,i}$. For example, the Zipf distribution ${p}_{\theta,i} = \frac{{1/i^\alpha}}{\sum_{j=1}^N {1/j^\alpha}}, i = 1,2,\ldots,N$ satisfies this property.

Let the true underlying parameter be $\Theta:= \{\Theta_1,\ldots,\Theta_d\}$. Note that $\Theta_j \in [a,b]$ for all $j=1,2,\ldots, d$. Let the BS observe $n_p$ (number of requests) i.i.d samples $(X_{t,1},X_{t,2},\ldots,X_{t,n_p}) \in \mathcal{X}_t^{n_p}$ drawn from the distribution $\mathcal{P}_\Theta$. Also, let $\hat {\Theta}_{n_p}^{(i)} := \left(\hat{\Theta}_{n_p,i,1}, \hat{\Theta}_{n_p,i,2},\ldots,\hat{\Theta}_{n_p,i,d} \right) \in \mathbb{R}^d$, $i=1,2,\ldots,n_p$ denote the estimate of $\Theta$, based on a single observation, i.e., $\hat {\Theta}_{n_p}^{(i)} = f(X_{t,i})$, $i=1,2,\ldots,n_p$, where $f:\mathcal{X}_t \rightarrow [a,b]^d$ is an unbiased estimator of $\Theta$. In the above, $n_p$ denotes the number of requests made by the users corresponding to the BS $z$ in a time interval of $[0,\tau]$. Since $f(\cdot)$ is an unbiased estimator of $\Theta$, we have $\mathbb{E} \left\{ \hat{\Theta}_{n_p}^{(i)} | \Theta\right\} = \Theta$ for all $i=1,2,\ldots, n_p$. The estimate of $\Theta$ using $n_p$ samples is obtained as follows:
\begin{equation} \label{eq:param_est}
\hat{\Theta}_{n_p} = \frac{1}{n_p} \sum_{j=1}^{n_p} \hat{\Theta}_{n_p}^{(j)}.
\end{equation}
Note that $\hat{\Theta}_{n_p} := (\hat{\Theta}_{n_p,1},\hat{\Theta}_{n_p,2},\ldots, \hat{\Theta}_{n_p,d}),$ is also an unbiased estimator of $\Theta$, i.e., $\mathbb{E} \left\{\hat{\Theta}_{n_p} | \Theta, n_p\right\} = \Theta$. The following theorem provides a bound on the time complexity for a family of parameterized popularity profile satisfying \textbf{Assumption 1}.

\begin{thrm}
For the family $\mathcal{P}_\Theta$ satisfying \textbf{Assumption 1}, and given the estimator $\hat{\Theta}_{n_p}$, $\hat{\mathcal{T}}^* \leq \mathcal{T}^* + \epsilon$ for every $\epsilon > 0$ with probability at least $1 - \delta$ if
\begin{equation}
\tau > \frac{1}{\lambda_r (1-e^{-\sigma^2})} \log \left(\frac{1}{1 - \frac{1}{\lambda_u \pi R^2} \log \frac{2 d}{\delta}}\right),
\label{eq:zipf_bound}
\end{equation}
for $\lambda_u > \frac{1}{\pi R^2} \log \frac{2 d}{\delta}$, otherwise $\tau = \infty$, where $\sigma^2 \triangleq \frac{2 \Omega^2}{d C^2 (b-a)^2}$ and $\Omega   \triangleq  \frac{R_0 \epsilon}{2 B}$.
\label{thm:pac_param_pop_distr}
\end{thrm}
\begin{proof}
See Appendix \ref{app:proof_parameteric_family}.
\end{proof}

From \eqref{eq:zipf_bound}, we see that the bound on the training time is independent of $N$, and from a scaling perspective, the training time scales with $d$, $\lambda_r$ and $\lambda_u$. This amounts to a significant improvement compared to the nonparametric model studied in the previous sections of this paper, where the training time is shown to scale as $N^2 \log N$. A natural extension is to utilize the knowledge obtained from users' interactions with a social community, namely, the source domain samples. In the next subsection, we analyze the time complexity bound employing the TL-based approach for popularity profiles modeled using a parametric family of distributions.

%\vspace{-0.2in}
\subsection{Transfer Learning for Parametric Models} \label{subsec:param_tl}
In this subsection, we derive a lower bound on the training time when the BS has access to the source domain samples along with the target domain samples. Let the source domain samples $(X_{s,1},X_{s,2},\ldots, X_{s,m}) \in \mathcal{X}_s^m$ drawn {\iid} from $\mathcal{P}_{\Theta_s}$, where $\Theta_s \in \mathbb{R}^d$. Further, as before, we assume that $\exists$ $f:\mathcal{X}_s \rightarrow \mathbb{R}^d$, an unbiased estimate of $\Theta_s$.  As before, let the BS observe $n_p$ i.i.d. target domain samples from $\mathcal{X}_t^{n_p}$ drawn from $\mathcal{P}_{\Theta}$. An estimate of $\Theta$ based on the available source and target domain samples is obtained as follows:
\begin{enumerate}[(i)]
\item Using the source domain samples an estimate of $\Theta_s$, denoted $\hat{\Theta}_s$, is obtained in manner similar to that of target domain parameter $\Theta$ as explained earlier in this section.
\item Using the target domain samples, an estimate of $\Theta$ denoted $\hat{\Theta}_t$ is obtained as in \eqref{eq:param_est}.
\item The two estimates are fused to get an estimate of $\Theta$ as $\hat{\Theta}_{\text{tl}} \triangleq \lambda \hat{\Theta}_t  + (1 - \lambda) \hat{\Theta}_s$, where $\lambda \in [0,1]$ will be described shortly.
\end{enumerate}

\begin{thrm} \label{thm:time_complexity_centralized_TL_convex_param}
For the family $\mathcal{P}_{\Theta^{'}}$ satisfying \textbf{Assumption 1}, and given the estimator $\hat{\Theta}_{\text{tl}} \triangleq \lambda \hat{\Theta}_t  + (1 - \lambda) \hat{\Theta}_s$, we have $\hat{\mathcal{T}}^* \leq \mathcal{T}^* + \epsilon$ for every $\epsilon > 0$ with a probability of at least $1 - \delta$ if the condition specified by \eqref{eq:thrm6cond} at the top of the next page is satisfied,
\begin{figure*}[!t]
\begin{eqnarray}
\tau \geq  \frac{1}{\lambda_r (1-e^{-\sigma_t^2})} \log \left(\frac{1}{1 - \frac{1}{\lambda_u \pi R^2} \left(\log \frac{2 d}{\delta} + \log \frac{1}{1 - \frac{2 d}{\delta} \exp\left\{-\frac{2m (\bar{\Omega} - \norm{\Theta - \Theta_s}_2)^2}{(b-a)^2}\right\}}\right)}\right) \label{eq:thrm6cond}
\end{eqnarray}
\hrulefill
\end{figure*}
for $\lambda_u >  \frac{1}{\pi R^2}\left(\log \frac{2 d}{\delta} + \log \frac{1}{1 - \frac{2 d}{\delta} \exp\left\{-\frac{2m (\bar{\Omega} - \norm{\Theta - \Theta_s}_2)^2}{(b-a)^2}\right\}}\right)$. This holds for all $D_t < \frac{\Omega}{C} - \lambda \bar G$ and $0<\lambda< \min\left\{\frac{\Omega}{C \bar G}, 1\right\}$. Here, $\bar{\Omega} \triangleq \frac{1}{\lambda}\left(\frac{\Omega}{C} - D_t\right)$, $\sigma_t^2 \triangleq \frac{2 D_t^2}{d (1-\lambda)^2 (b-a)^2}$, $\Omega \triangleq \frac{\tilde{\epsilon}}{\sup_i g(\pi_i)}$, $\tilde \epsilon \triangleq  \frac{R_0 \epsilon}{2 B}$, and $\bar G : = \norm{\Theta - \Theta_s}_2 + (b-a) \sqrt{\frac{1}{2m} \log\frac{2 d}{\delta}}$.
\end{thrm}
\begin{proof}
%\emph{Proof:}
See Appendix \ref{app:estimation_TL_convex_param}.
\end{proof}
It is important to note that the aformentioned bound is independent of $N$. In the following section, we provide numerical results to get further insights into the expressions derived in the paper.
%In the following section, insights on the proposed TL based approach in comparison to the source sample agnostic method is provided.

\section{Numerical Results} \label{sec:sims}
In this section, we provide numerical results and derive insights into the analyses carried out in the previous sections. The parameter values used in our calculations are as follows: $B = 10^7~\text{bits}$, $R_0 = 10^6~\text{bits/s}$, $\gamma = 100$m, $\lambda_u = 0.001~\text{nodes}/m^2$, $\lambda_r = 1/360~\text{requests/s}$, $\lambda_s= 10^{-5}~\text{nodes}/m^2$, $\delta = 0.02$, $R=2$ Km, $m = 10^5$ samples, $\beta_e = 0.6$ and $\beta_l = 0.2$. $\epsilon$ is chosen as a fraction of a lower bound on the offloading loss, {\ie}, $\mathcal{T}(\Pi, \mathcal{P}) \geq \frac{B}{R_0} \exp\{-\lambda_s \pi \gamma^2\}$. In particular, $\epsilon = \text{fraction} \times \frac{B}{R_0} \exp\{-\lambda_s \pi \gamma^2\}$. Further, $\norm{\mathcal{P}-\mathcal{Q}}_\infty = 0.1 \left(\frac{\epsilon R_0}{2 B N}\right)$ which for the above parameters is of the order of $10/N$.

\begin{figure}
\begin{center}
{\includegraphics[height=6cm,width=9.0cm]{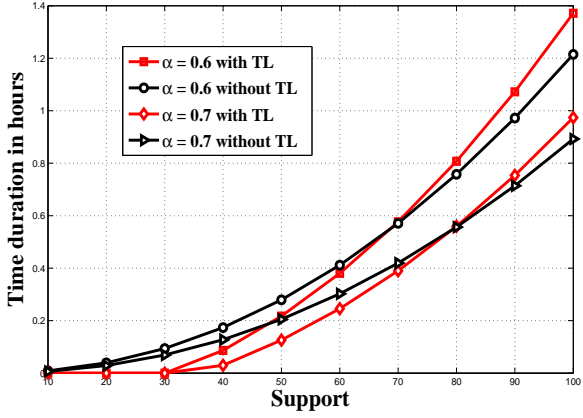}}
\caption{Training duration versus $N$, corresponding to Theorems \ref{thm:lower_bound_nonparam_ntl} and \ref{thm:lower_bound_nonparam_tl}.}
\label{fig:fig1}
\end{center}
\end{figure}

\figref{fig:fig1} shows a plot of the lower bounds on the training duration obtained in Theorems \ref{thm:lower_bound_nonparam_ntl} and \ref{thm:lower_bound_nonparam_tl} as functions of the support $N$. It is seen that, for $N \leq 70$ the TL-based approach provides significant performance improvement. However, for $N > 70$, the performance of the TL-based approach degrades compared to the approach that uses only the source domain samples (and, hence, can be called agnostic). This suggests that for larger values of $N$, the estimate of the popularity profile obtained using \eqref{eq:TL_convex_comb_estimate} performs poorly due to incorrect fusion of the estimates obtained from source and target domains.
\begin{figure}
\begin{center}
{\includegraphics[height=6cm,width=9.0cm]{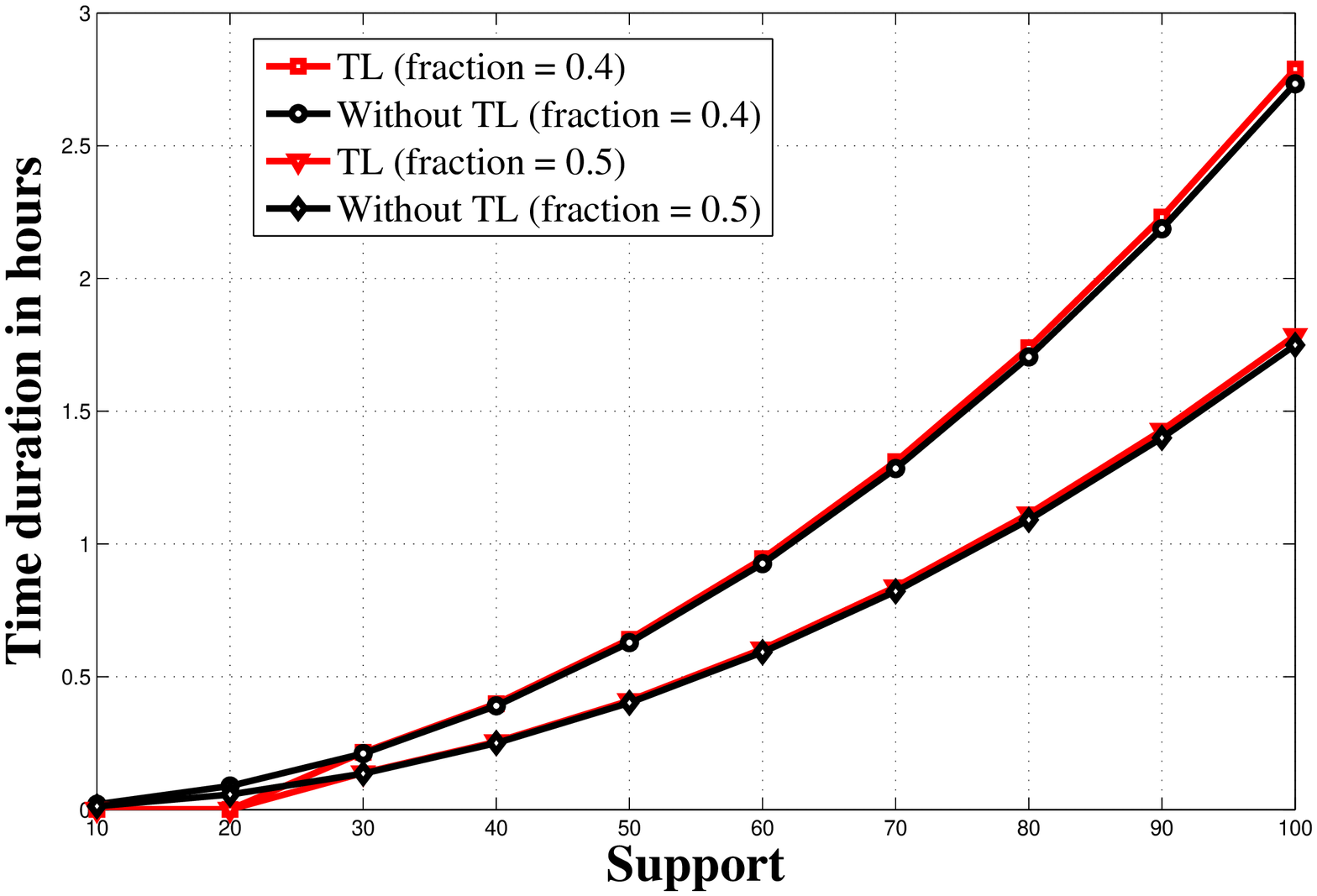}}
\caption{Training duration versus $N$, corresponding to \thrmref{thm:conv_comb}.}
\label{fig:fig2}
\end{center}
\end{figure}
Fig. \ref{fig:fig2} shows the plots of the lower bound in \thrmref{thm:conv_comb} corresponding to the estimate obtained by a fixed linear combination of the source and target estimates (see \eqref{eq:TL_convex_comb_estimate}). As seen in the figure, this does not bring any performance improvement and in fact sometimes performs poorly compared to the source domain agnostic approach. This is because the fixed linear combination does not have the flexibility to adapt to different realizations of the network, proving the sub-optimality of the estimate in \eqref{eq:TL_convex_comb_estimate} compared to that in \eqref{eq:tl_flexible_estimate}. It is also seen that the coefficients used in the estimate that adapts to the varying realizations of the network as in \eqref{eq:TL_convex_comb_estimate} is beneficial.

%\vspace{-0.3in}
\begin{figure}
\begin{center}
{\includegraphics[height=6cm,width=9.0cm]{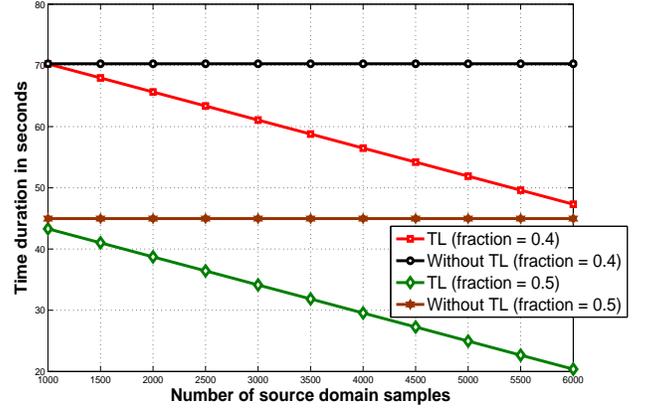}}
\captionof{figure}{Training duration versus $m$ for a fixed $N (= 10)$.}
\label{fig:fig3}
\end{center}
\end{figure}

\begin{figure}
\begin{center}
{\includegraphics[height=6cm,width=9.0cm]{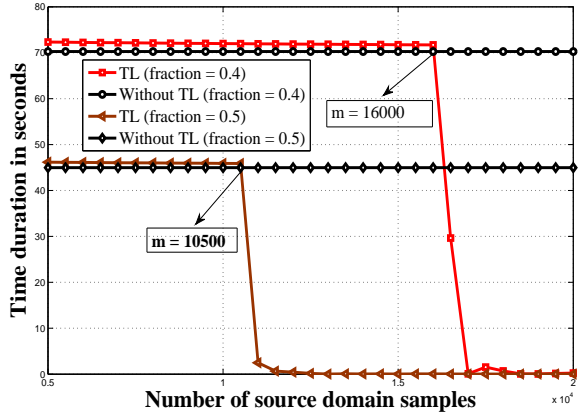}}
\captionof{figure}{Training duration versus $m$ for a fixed $N (= 10)$. Comparing TL-based and agnostic approaches.}
\label{fig:fig4}
\end{center}
\end{figure}

\figref{fig:fig3} shows a plot similar to that in \figref{fig:fig1} but with $N=10$ and varying $m$. It can be seen that the TL-based approach performs better for all $m\geq 1000$ demonstrating its applicability in practice. As seen, the performance is better for higher values of the fraction which corroborates intuition. \figref{fig:fig4} also shows a plot of time duration versus $m$ for a fixed $N=10$. It can be seen that the estimate in \eqref{eq:TL_convex_comb_estimate} outperforms the agnostic approach; however, this is observed at very high values of source domain samples ($m = 10500$ and $m=16000$ for $\text{fraction} = 0.5$ and $\text{fraction} = 0.4$, respectively). Thus, although the TL-based approach using the estimate \eqref{eq:TL_convex_comb_estimate} has some benefits, it is not desirable for practical applications.

\begin{figure}[h!]
\begin{center}
{\includegraphics[height=6cm,width=9.0cm]{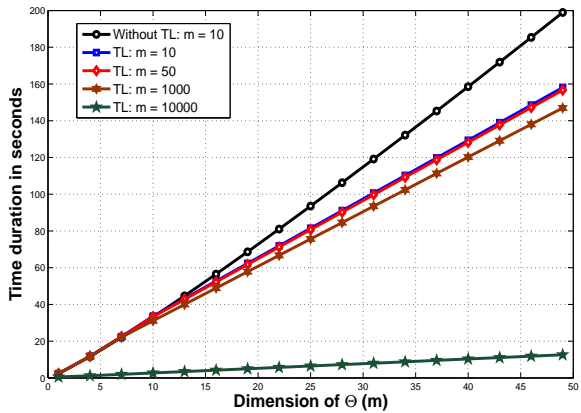}}
\caption{Training duration versus $\Theta$ when the popularity profile is modeled using a parametric family of distributions. $C = 2,~\text{fraction} = 0.6,~\norm{\Theta - \Theta_s}_2 = 0.1,~(b-a)=0.5$.}
\label{fig:fig5}
\end{center}
\end{figure}
The main benefits of the TL-based approach are shown in \figref{fig:fig5} for the parametric family of popularity profiles. It can be seen that the TL-based approach performs significantly better than the source domain agnostic approach for values of $m$ as low as $10$. This is because the number of parameters to be estimated scales with the dimension of $\Theta$ rather than with the support. In particular, as $d$ increases the training duration also increases, which is quite expected. However, the delay scales only linearly in $d$ as compared to quadratic scaling experienced with the nonparametric method.

%\vspace{-0.3in}
\section{Concluding Remarks}\label{sec:conclude}
The popularity profile for caching in distributed heterogenous cellular networks was estimated at BS using the available instantaneous demands from users in a time interval $[0,\tau]$. We showed that a training time $\tau$ to achieve an $\epsilon>0$ difference between the achieved cost and the optimal cost was finite, provided the user density was greater than a threshold; $\tau$ was shown to scale as square of the support of the popularity profile. A TL-based approach was proposed to estimate the popularity profile, and a condition was derived under which it performed better than the target domain sample only based approach. Although TL-based approach performs better, the error that is achieved in \eqref{eq:epsilon_error_bound} depends on $\norm{\mathcal{P}-\mathcal{Q}}_\infty$, suggesting that lower the distance between the two distributions better the TL scheme performs. From Proposition \ref{prop:TL_cache}, the benefits of using target domain samples can only be realized with the knowledge of the distance $\norm{\mathcal{P}-\mathcal{Q}}_\infty$. The main benefit of the TL-based approach is recognized when the popularity profile is modeled using a parametric family of distributions. In this case, the delay is independent of $N$ and scales only linearly with the dimension of the distribution parameter. In practice, caching depends on several factors such as the scheduling scheme used, which in turn depends on the channel conditions, QoS requirements, {\etc}. An important assumption that we make is that if the requested file is present in one (or more) of the neighboring SBSs, the transmissions are scheduled within a tolerable time frame. In the case of caching, this time duration could be slightly relaxed, and can be thought of as an abstraction of the scheduling scheme employed. If the file is not present, regardless of the scheduling policy, the file cannot be served locally. Hence, the approach that we have leads to a lower bound, albeit pessimistic, on the training time. Thus, even under pessimistic situations, the training time scales as $N^2\log N$ for achieving an offloading loss that is $\epsilon > 0$ away from the optimal offloading loss.

\vspace{-0.1in}
\section*{Acknowledgement}
K. G. Nagananda would like to thank Chandra R. Murthy, at the Indian Institute of Science, for providing the lab space during the course of this work. The work of H. Vincent Poor was supported in part by the U. S. National Science Foundation under Grant CNS-1456793. The authors thank the anonymous referees for their comments and suggestions.

%\vspace{-0.2in}

\bibliographystyle{IEEEtran}
\bibliography{IEEEabrv,caching2015}
\raggedbottom

\appendices
%

%\vspace{-0.2in}

\section{Proof of \thrmref{thm_mean_throughput}} \label{app:througput_derivation}
The first term in \eqref{eq:metric}, $\mathbb{E} \sum_{i=1}^N\mathbf{1} \{f_i \notin \mathcal{N}_o\} \mathbf{1} \{f_i \text{ requested}\}$, can be written as
\begin{eqnarray*} \label{eq:mean_througput_expression_1term}
&& \hspace{-2mm} \mathbb{E}_{n_s} \sum_{i=1}^N \mathbb{E} \mathbf{1}\{f_i \text{ requested}\}\Pr\{f_i \notin \mathcal{N}_o \vert \abs{\mathcal{N}_o} = n_s\} \\
\hspace{-2mm}&\stackrel{(a)}{=}&\hspace{-2mm} \mathbb{E}_{n_s}\sum_{i=1}^N\left[\Pr\{f_i \notin a,\text{ for \textbf{any} } a \in \mathcal{N}_o\}\right]^{n_s} p_i  \\
&\stackrel{(b)}{=}& \sum_{i=1}^N\mathbb{E}_{n_s}(1-\pi_i)^{n_s M} p_i \\
&\stackrel{(c)}{=}&\mathbb{E} \sum_{i=1}^N \sum_{j=0}^\infty (1-\pi_i)^{j M} e^{\{-\lambda_s \pi \gamma^2\}} \frac{(\lambda_s \pi \gamma^2)^j}{{j !}} \\
&=& \sum_{i=1}^N  \exp\left\{-\mathcal{U}\right\}  p_i,
\end{eqnarray*}
where $\mathcal{U}\triangleq\lambda_s \pi \gamma^2 \left[1-(1-\pi_i)^M\right]$. In the above exposition, $(a)$ follows from the fact that the proposed random caching scheme is independent across users, $(b)$ is due to the fixed cache size ($M$), and $(c)$ follows since $n_s$ is a PPP with mean $\lambda_s \pi \gamma^2$, where $n_s$ is the number of SBSs in a circular area of radius $\gamma$. This completes the proof of \thrmref{thm_mean_throughput}. $\blacksquare$

%\vspace{-0.151in}
\section{Proof of \thrmref{thm:time_complexity_centralized}} \label{app:time_complexity_centralized}
For any $\epsilon>0$, the inequality $\Pr\{\hat{\mathcal{T}}^* \geq \mathcal{T}^* + \epsilon\} \leq \Pr\left\{2 \sup_{\mathbf{1} \succeq \Pi \succeq \mathbf{0}:\mathbf{1}^T \Pi = 1} \abs{\Delta \mathcal{T}} > \epsilon\}\right\}$ is proved, where $\Delta \mathcal{T}\triangleq \mathcal{T}(\Pi,\hat{\mathcal{P}}) - \mathcal{T}(\Pi,\mathcal{P})$. And, $\hat{\mathcal{T}}^*- \mathcal{T}^*$ can be written as (see \cite{vapnik_Chervonenkis1974c})
\begin{eqnarray*}
\hat{\mathcal{T}}^* - \inf_{\Pi} \mathcal{T}(\Pi,{\mathcal{P}})
\leq   \hat{\mathcal{T}^*} -  \hat{\mathcal{T}}  + \sup_{\Pi} \abs{\mathcal{T}(\Pi,\hat{\mathcal{P}}) - \mathcal{T}(\Pi,{\mathcal{P}})} \\
\leq 2\sup_{\Pi} \abs{\mathcal{T}(\Pi,\hat{\mathcal{P}}) - \mathcal{T}(\Pi,{\mathcal{P}})},
\end{eqnarray*}
where $\hat{\mathcal{T}}\triangleq\mathcal{T}(\Pi,\hat{\mathcal{P}})$, thus proving the inequality. Substituting for $\mathcal{T}(\Pi,\hat{\mathcal{P}})$ and $\mathcal{T}(\Pi,{\mathcal{P}})$ from \eqref{eq:mean_througput_expression} we get $\Pr \left\{ \sup_\Pi  \abs{\sum_{i=1}^N g(\pi_i) (\hat{p}_i - p_i)} > \tilde{\epsilon} \right\} $, which can be upper bounded as follows:
\begin{eqnarray}
\nonumber \Pr \left\{ \sup_\Pi {\sum_{i=1}^N g(\pi_i) \hat{\delta}_{p,i}} > \tilde{\epsilon} \right\} \!\!\!
&\leq& \!\!\! \Pr \left\{{ \max_{i=1,2,\ldots,N} \hat{\delta}_{p,i}} > \bar{\epsilon} \right\} \\
\nonumber &\leq& \!\!\! \sum_{i=1}^N \Pr \left\{{ \hat{\delta}_{p,i}} > \bar{\epsilon} \right\} \\
&\leq& \!\!\! 2N \mathbb{E}\left[\exp\left\{-  {2 {\bar{\epsilon}}^2 n_p} \right\}\right], \label{eq:hoeffdings_baseline}
\end{eqnarray}
where $\hat{\delta}_{p,i} \triangleq \abs{\hat{p}_i - p_i}$, $\tilde \epsilon \triangleq  \frac{R_0 \epsilon}{2 B}$, $\bar \epsilon \triangleq \frac{\tilde{\epsilon}}{\sup_\Pi \left\{\sum_{i=1}^N g(\pi_i)\right\} }$, and $g(\pi_i)\triangleq \exp\{-\lambda_s \pi \gamma^2 \left[1 - (1-\pi_i)^M \right]\}$, and the last inequality follows by applying Hoeffdings inequality (see \cite{Devroye2014}) since the estimator $\hat{\mathcal{P}}$ is unbiased and $\pi_i$, $\pi \in [0,1]$. Note that, the expectation in \eqref{eq:hoeffdings_baseline} is with respect to $n_p$. Conditioned on the number $n_R$ of users in the coverage area of BS, $n_p$ is a Poisson distributed random variable with density $n_R \lambda_r \tau$. Therefore, $2 N \mathbb{E} \sum_{n=0}^\infty \exp\left\{-\bar g\right\} \frac{(\lambda_r n_R \tau)^n}{n!} = 2 N \mathbb{E}_{n_R} \exp\{-\lambda_r n_R \tau g^*\}$, where $\bar g \triangleq \left({2 \bar{\epsilon}^2 k } + \lambda_r n_R \tau\right)$ and $g^* \triangleq \left( 1 - \exp\left\{-{2 \bar{\epsilon}^2 }\right\}\right)$ which can further be simplified as
\begin{eqnarray}
\nonumber 2 N \sum_{k=0}^\infty\!\!  \exp\{-\lambda_r k \tau  g^*\} \exp\{-\lambda_u \pi R^2\} \frac{(\lambda_u \pi R^2)^k}{k!}\\ =  2 N\exp\{-\lambda_u \pi R^2 \left(1- \exp\left\{-\lambda_r \tau g^*\right\} \right).
\label{eq:error_bound}
\end{eqnarray}
We see that $\Pr\left\{ \sup_{\Pi} \abs{\Delta \mathcal{T}} > \frac{\epsilon}{2}\right\} \leq \delta$ if \eqref{eq:error_bound} is upper bounded by $\delta$, resulting in
\begin{eqnarray*}
\tau \geq \frac{1}{\lambda_r g^*} \log \left(\frac{1}{1- \frac{1}{\lambda_u \pi R^2} \log \frac{2 N}{\delta}}\right),
\end{eqnarray*}
provided $\lambda_u > \frac{1}{\pi R^2} \log \frac{2 N}{\delta}$,  otherwise $\tau = \infty$, proving \thrmref{thm:time_complexity_centralized}. $\blacksquare$

%\vspace{-0.2in}
\section{Proof of \thrmref{thm:time_complexity_centralized_TL}} \label{app:TL_time_complexity}
It is easy to see that $\Pr\{\hat{\mathcal{T}}^* \geq \mathcal{T}^* + \epsilon\} \leq \Pr \left\{ \sup_{1\leq i \leq N} \abs{\hat{p}_i^{(tl)} - p_i} > \bar{\epsilon}\right\}$, where $\bar \epsilon \triangleq \frac{R_0 {\epsilon}}{2 B \sup_\Pi \left\{\sum_{i=1}^N g(\pi_i)\right\} }$ and $g(\pi_i)\triangleq\exp\{-\lambda_s \pi \gamma^2 \left[1 - (1-\pi_i)^M \right]\}$. Denote by $n_{p}$ the total number of requests in the coverage area of the BS. Conditioned on the number of users $n_R$ in the coverage area of the BS, $n_{p}$ is a Poisson distributed random variable with density $n_R \lambda_r$. Further, $\mathbb{E}\left\{ \hat{p}_i^{(tl)} | n_p\right\} = \frac{n_{p}}{n_{p} + m}p_i + \frac{m}{n_{p} + m}q_i$. Using this, we can write
\begin{eqnarray*}
&& \Pr \left\{ \sup_{1\leq i \leq N} \abs{\hat{p}_i^{(tl)}  - \mathbb{E} \hat{p}_i^{(tl)} + \mathbb{E} \hat{p}_i^{(tl)} - p_i} > \bar{\epsilon} \right\} \!\!\!\!\! \\
&\leq& \!\!\!\! \Pr \left\{ \sup_{1\leq i \leq N} \abs{\hat{p}_i^{(tl)}  - \mathbb{E} \hat{p}_i^{(tl)}} + \abs{\mathbb{E} \hat{p}_i^{(tl)} - p_i} > \bar{\epsilon} \right\}\\
&\leq& \!\!\!\!\! \Pr \left\{ \sup_{1\leq i \leq N} \abs{\hat{p}_i^{(tl)}  - \mathbb{E} \hat{p}_i^{(tl)}}  > \bar{\epsilon} -  \sup_{i\in[1,N]} \abs{\mathbb{E} \hat{p}_i^{(tl)} - p_i}\right\} \nonumber\\
&\leq& \!\!\!\!\! \Pr \left\{ \sup_{1\leq i \leq N} \abs{\hat{p}_i^{(tl)}  - \mathbb{E} \hat{p}_i^{(tl)}}  > \bar{\epsilon} -  \frac{m}{n_{p} + m} \norm{\mathcal{P} - \mathcal{Q}}_{\infty} \right\} \nonumber\\
&\leq& \!\!\! \mathbb{E}_{n_{p}}\Pr \left\{ \sup_{1\leq i \leq N} \abs{\hat{p}_i^{(tl)}  - \mathbb{E} \hat{p}_i^{(tl)}}  > \bar{\epsilon} -   \norm{\mathcal{P} - \mathcal{Q}}_{\infty} \left | \right. n_{p} \right\}\nonumber \\
&\leq& \!\!\! N \mathbb{E}_{n_{p}}\Pr \left\{\abs{\hat{p}_i^{(tl)}  - \mathbb{E} \hat{p}_i^{(tl)}}  > \bar{\epsilon} -   \norm{\mathcal{P} - \mathcal{Q}}_{\infty} \left | \right. n_{p} \right\},\nonumber
\end{eqnarray*}
provided $\bar \epsilon >  \norm{\mathcal{P} - \mathcal{Q}}_{\infty}$, where $\norm{\mathcal{P} - \mathcal{Q}}_{\infty} \triangleq \sup_{i\in[1,N]} \abs{q_i - p_i}$. From Hoeffding's inequality,
\begin{eqnarray}
&& \nonumber 2 N \mathbb{E}_{n_{p}} \exp\left\{-2 \epsilon_{pq}^2 (n_{p} + m) \right\} \!\!\! \\ &=& \!\! 2N \mathbb{E}_{n_R} \exp\left\{-(2 \epsilon_{pq}^2 m +\lambda_r n_R \tau)\right\} \sum_{k=0}^\infty \frac{a^k}{k!},\nonumber \\
 &=& \!\!\! 2 N \exp\left\{-2 \epsilon_{pq}^2 m\right\} \mathbb{E}_{n_R} \exp\{-\bar{g}_{pq}\} \nonumber \\
 &=& \!\!\! 2 N \exp\left\{-2 \epsilon_{pq}^2 m\right\} \exp\{-\lambda_u \pi R^2\} \times \nonumber\\
 && \sum_{l=0}^\infty \exp\{- \lambda_r l \tau(1-\exp\{-2 \epsilon_{pq}^2\})\} \frac{(\lambda_u \pi R^2)^l}{l!},
 \end{eqnarray}
where $a \triangleq \lambda_r n_R \tau \exp\left\{-2 \epsilon_{pq}^2 \right\} $, $\epsilon_{pq} \triangleq \bar \epsilon - \norm{\mathcal{P} - \mathcal{Q}}_{\infty}$ and $\bar{g}_{pq}\triangleq \lambda_r n_R \tau(1-\exp\{-2\epsilon_{pq}^2\}) $. Therefore, $2 N \mathbb{E}_{n_{p}} \exp\left\{-2 \epsilon_{pq}^2 (n_{p} + m) \right\} = 2 \exp\left\{-2 \epsilon_{pq}^2 m\right\} \exp\{-\lambda_u \pi R^2 t\}$, where \\ $t \triangleq \left(1 - \exp\{-\lambda_r  \tau\left(1 - \exp\{-2 \epsilon_{pq}^2\} \right)\}\right)$, and is at most $\delta >0$ if
\begin{eqnarray}
\tau \geq \frac{1}{\lambda_r (1-e^{-2 \epsilon_{pq}^2})} \log \left(\frac{1}{1 - \frac{1}{\lambda_u \pi R^2} \left(\log \frac{2N}{\delta} - 2 \epsilon_{pq}^2 m\right)} \right),
\end{eqnarray}
provided $\lambda_u >  \frac{1}{\pi R^2} \left(\log \frac{2N}{\delta} - 2 \epsilon_{pq}^2 m\right)$, otherwise $\tau = \infty$, thus proving \thrmref{thm:time_complexity_centralized_TL}. $\blacksquare$

%\vspace{-0.1in}
\section{Proof of \thrmref{thm:time_complexity_centralized_TL_convex}} \label{app:estimation_TL_convex}
\begin{figure*}
\begin{eqnarray}
\Pr \left\{ \abs{\hat{p}_i^{(tl)} - p_i} > \bar{\epsilon} \right\} &=& \Pr \left\{ \abs{\alpha \hat{p}_i^{(s)} + (1 - \alpha) \hat{p}_i^{(t)})- p_i} > \bar{\epsilon} \right\} \label{eq:target_source_bound_tl_convcomb_k1} \\
&\leq& \Pr \left\{ \alpha \abs{(\hat{p}_i^{(s)} - p_i)}+ (1 - \alpha) \abs{(\hat{p}_i^{(t)}- p_i)} > \bar{\epsilon} \right\} \label{eq:target_source_bound_tl_convcomb_k2} \\
&=& \Pr \left\{ \left[\alpha \abs{\hat{p}_i^{(s)} - p_i}+ (1 - \alpha) \abs{\hat{p}_i^{(t)}- p_i} > \bar{\epsilon}\right]  \bigcap  \abs{\hat{p}_i^{(t)}- p_i} > \eta \right\} \\ && + \Pr \left\{ \alpha \abs{\hat{p}_i^{(s)} - p_i} + (1 - \alpha) \abs{\hat{p}_i^{(t)}- p_i} > \bar{\epsilon}  \bigcap  \abs{\hat{p}_i^{(t)}- p_i} \leq \eta \right\}\nonumber \label{eq:target_source_bound_tl_convcomb_k3}\\
&{\leq}&   \Pr \left\{  \abs{\hat{p}_i^{(t)} - p_i}  > \eta \right\} + \Pr \left\{  \abs{\hat{p}_i^{(s)} - p_i}  >\omega  \right\},\label{eq:target_source_bound_tl_convcomb_k4}
\end{eqnarray}
\hrulefill
\end{figure*}
We begin with $\Pr\{\hat{\mathcal{T}}^* \geq \mathcal{T}^* + \epsilon\} \leq \Pr \left\{ \sup_{1\leq i \leq N} \abs{\hat{p}_i^{(tl)} - p_i} > \bar{\epsilon} \right\} \leq \sum_{i=1}^N \Pr \left\{ \abs{\hat{p}_i^{(tl)} - p_i} > \bar{\epsilon} \right\}$,
where $\bar \epsilon \triangleq \frac{R_0 {\epsilon}}{2 B \sup_\Pi \left\{\sum_{i=1}^N g(\pi_i)\right\} }$ and $g(\pi_i)\triangleq \exp\{-\lambda_s \pi \gamma^2 \left[1 - (1-\pi_i)^M \right]\}$. Each term in the summation can be upper bounded as shown in \eqref{eq:target_source_bound_tl_convcomb_k1} - \eqref{eq:target_source_bound_tl_convcomb_k4} at the top of the next page, where $\omega \triangleq \frac{\bar{\epsilon} - (1 - \alpha) \eta}{\alpha} > 0$. From \eqref{eq:error_bound}, we have
$\Pr \left\{  \abs{\hat{p}_i^{(t)} - p_i}  > \eta \right\} = 2 \exp\{-\lambda_u \pi R^2 \left(1- \exp\left\{-\lambda_r \tau g_t^*\right\} \right)$, where $g_t^* \triangleq \left( 1 - \exp\left\{-{2 \eta^2 }\right\}\right)$ and the second term can be bounded as follows:
\begin{eqnarray*}
\Pr \left\{  \abs{\hat{p}_i^{(s)} - p_i}  > \omega \right\} \!\!\! &=& \!\!\! \Pr \left\{  \abs{\hat{p}_i^{(s)} - q_i + q_i - p_i}  > \omega \right\} \\
&\stackrel{(a)}{\leq}& \!\!\! \Pr \left\{  \abs{\hat{p}_i^{(s)} - q_i}   > \omega - \norm{\mathcal{P} - \mathcal{Q}}_\infty\right\},
\end{eqnarray*}
where $(a)$ follows from the triangular inequality and using $\norm{\mathcal{P} - \mathcal{Q}}_\infty = \sup_{i}\abs{p_i - q_i}$. Note that, the inequality (a) is valid if $\omega > \norm{\mathcal{P} - \mathcal{Q}}_\infty$. Using $\mathbb{E} \hat{p}_i = p_i$ and Hoeffding's inequality, we have $\Pr \left\{  \abs{\hat{p}_i^{(s)} - q_i}   > \omega - \norm{\mathcal{P} - \mathcal{Q}}_\infty\right\} \leq 2 \exp\left\{-2 (\omega - \norm{\mathcal{P} - \mathcal{Q}}_\infty)^2 m\right\}$. Therefore,
\begin{eqnarray}
\nonumber \sum_{i=1}^N \Pr \left\{ \abs{\hat{p}_i^{(tl)} - p_i} > \bar{\epsilon} \right\} \leq 2N[\exp\left\{-2 {\bar \omega}^2 m\right\} \\ + \exp\{-\lambda_u \pi R^2 \left(1- \exp\left\{-\lambda_r \tau g_t^*\right\} \right)],
\label{eq:prob_appendD}
\end{eqnarray}
where $\bar \omega = \omega - \norm{\mathcal{P} - \mathcal{Q}}_\infty > 0$.
%The constraint $\omega > \norm{\mathcal{P} - \mathcal{Q}}_\infty$ leads to $0 < \alpha < \min \left\{\frac{\bar \epsilon}{\norm{\mathcal{P} - \mathcal{Q}}_\infty},1\right\}$ and $\eta < \frac{\bar \epsilon - \alpha  \norm{\mathcal{P} - \mathcal{Q}}_\infty}{1-\alpha}$.
Finally, it is clear that \eqref{eq:prob_appendD} can upper bounded by $\delta$ provided
\begin{equation}
\tau \geq  \frac{1}{\lambda_r g_t^*} \log \left[\frac{1}{1 - \frac{1}{\lambda_u \pi R^2 } \left[\log\left(A_1\right) + \log\left(B_1\right) \right]}\right],
\end{equation}
where $A_1 =  \frac{2N}{\delta}$, $B_1 = \frac{1}{1 - \left(\frac{2N}{\delta}\right)\exp\{-2{\bar \omega}^2 m\}}$, and $\lambda_u >  \frac{1}{\pi R^2 } \left(\log\left(\frac{2N}{\delta}\right) + \log\left\{\frac{1}{1 - \left(\frac{2N}{\delta}\right)\exp\{-2{\bar \omega}^2 m\}}\right\} \right)$, which is valid if $\left(\frac{2N}{\delta}\right)\exp\{-2{\bar \omega}^2 m\} < 1$. This along with $\omega - \norm{\mathcal{P} - \mathcal{Q}}_\infty > 0$ leads to the constraint stated in \thrmref{thm:time_complexity_centralized_TL_convex}. $\blacksquare$

%\vspace{-0.1in}
\section{Proof of \thrmref{thm:pac_param_pop_distr}} \label{app:proof_parameteric_family}
We begin with
\begin{eqnarray*} %\label{eq:tl_param_error_bound}
&& \Pr\{\hat{\mathcal{T}}^* \geq \mathcal{T}^* + \epsilon\} \\ &\leq& \Pr \left\{ \sup_{0 \leq \pi \leq 1} g(\pi)\sum_{i=1}^N \abs{{p}_{{\hat{\Theta}_{n_p},i}} - p_{\Theta,i}} > \Omega \right\}
\\ &\stackrel{(a)}{\leq}&  \Pr \left\{ \sum_{i=1}^N\abs{{p}_{{{\hat \Theta}_{n_p}},i} - p_{\Theta,i}} > \Omega \right\},
\end{eqnarray*}
where $(a)$ follows from the fact that $\sup_{0 \leq \pi \leq 1} g(\pi) = 1$, ${p}_{{\hat \Theta}_{n_p},i}$ is the estimate of ${p}_{{\Theta},i}$, $\Omega \triangleq \frac{R_0 \epsilon}{2 B}$ and $g(\pi)\triangleq \exp\{-\lambda_u \pi \gamma^2 \left[1 - (1-\pi_i)^M \right]\}$. By using the remainder form of the Taylor series, ${p}_{{\hat \Theta}_{n_p},i} = p_{\Theta,i} + (\Theta - \hat{\Theta}_{n_p})  \partial p_{\Theta^*,i} \big |_{\Theta^* \in [\Theta, {\hat \Theta}_{n_p}]}$, where $[\Theta, {\hat \Theta}_{n_p}]$ represents the line joining the points $\Theta$ and ${\hat \Theta}_{n_p}$, leading to (recall that the $i$-th component of $\hat{\Theta}_{n_p}$ is denoted by $\hat{\Theta}_{n_p,i}$, $i = 1,2,\ldots,d$)
\begin{eqnarray}
\nonumber && \Pr \left\{ \sum_{i=1}^N  \abs{{p}_{{\hat \Theta}_{n_p},i} - p_{\Theta,i}} > \Omega \right\}\!\!\! \\ \nonumber &\stackrel{(a)}{\leq}& \!\!\! \Pr\left\{\norm{{\hat \Theta}_{n_p} - \Theta}_2 \sup_{\Theta^* \in [\Theta, {\hat \Theta}_{n_p}]}\sum_{i=1}^N\norm{ \partial p_{\Theta^*,i} }_2 > \Omega \right\}\\
\nonumber  &\leq&\!\!\! \Pr\left\{\norm{{\hat \Theta}_{n_p} - \Theta}_2^2  > \frac{{\Omega}^2}{C^2} \right\}\\
\nonumber  &\stackrel{(b)}{\leq}&\!\!\!  \Pr\left\{\sup_{1\leq i \leq d} \abs{{\hat \Theta}_{n_p,i} - \Theta_i}^2  > \frac{{\Omega}^2}{d C^2} \right\}\\
&\leq& d \Pr\left\{\abs{{\hat \Theta}_{n_p,i} - \Theta_i}^2  > \frac{{\Omega}^2}{d C^2} \right\},
\end{eqnarray}
where $(a)$ follows from the Cauchy-Schwartz inequality and \textbf{Assumption 1} in \secref{sec:param_pac_bound}, and $(b)$ follows from the fact that $\norm{\hat \Theta_{n_p} - \Theta}_2^2 = \sum_{i=1}^d \abs{\hat{\Theta}_{n_p,i} - \Theta_i}^2 \leq d \sup_{1\leq i \leq d} \abs{\hat{\Theta}_{n_p,i} - \Theta_i}^2$.

First, note that for all $i$, $\hat{\Theta}_{n_p,i}$ is an unbiased estimate of $\Theta_i$. Further, $a \leq \Theta_j \leq b$ for $j=1,2,\ldots,d$. Thus, by applying Hoeffding's inequality, we have
\begin{equation} \label{eq:thetaij_bound}
d \Pr\left\{\abs{\hat{\Theta}_{n_p,i} - \Theta_i}^2  > \frac{{\Omega}^2}{d C^2} \right\} \leq 2 d \mathbb{E}_{n_p}\exp\left\{-n_p \sigma^2 \right\},
\end{equation}
where $\sigma^2 \triangleq \frac{2 \Omega^2}{d C^2 (b-a)^2}$. Conditioned on the number of users (denoted $n_R$) in a radius of $R$ around the BS, $n_p$ is PPP with density $\lambda_r \tau n_R$. Using this fact in \eqref{eq:thetaij_bound}, we can write
\begin{eqnarray}
\nonumber &&  d \Pr\left\{\abs{\hat{\Theta}_{n_p,i} - \Theta_i}^2  > \frac{{\Omega}^2}{d C^2} \right\}\!\!\!\! \\ \nonumber &\leq& \!\!\!\! 2 d \mathbb{E}_{n_R}  \left\{ \mathbb{E}_{n_p} \left[ \exp\left\{-n_p \sigma^2 \right\} | n_p\right]\right\} \\
\nonumber &=&\!\!\!\!  2 d \mathbb{E}_{n_R}   \sum_{k=0}^\infty  \exp\{-k \sigma^2\} e^{-n_R \lambda_r \tau} \frac{(n_R \lambda_r \tau)^k}{k!} \\
\nonumber &=&\!\!\!\! 2 d \mathbb{E}_{n_R}  e^{-n_R \lambda_r \tau} \sum_{k=0}^\infty   \frac{(n_R \lambda_r \tau e^{-\sigma^2})^k}{k!}
\\ \nonumber &=& 2 d \mathbb{E}_{n_R}   \exp\{-n_R \lambda_r \tau (1 + e^{-\sigma^2})\}\\
\nonumber &=&\!\!\!\! 2 d    \sum_{k=0}^\infty   \exp\{-k \lambda_r \tau (1 + e^{-\sigma^2})\} \frac{(\lambda_u \pi R^2)^k}{k!} e^{-\lambda_u \pi R^2}
\\ &=& f_{\sigma^2}(\tau), \label{eq:fsigma_ntl_param}
\end{eqnarray}
where $f_{\sigma^2}(\tau) \triangleq 2 d \exp\left\{-\lambda_u \pi R^2 \left(1 - \exp\{-\lambda_r \tau (1-e^{-\sigma^2})\}\right)\right\}$ is a monotonically decreasing function of $\tau$ for all $\tau > 0$. Thus, $f_{\sigma^2}(\tau) \leq \delta$ if $\tau > \frac{1}{\lambda_r (1-e^{-\sigma^2})} \log \left(\frac{1}{1 - \frac{1}{\lambda_u \pi R^2} \log \frac{2 d}{\delta}}\right)$, for $\lambda_u > \frac{1}{\pi R^2} \log \frac{2 d}{\delta}$, proving \thrmref{thm:pac_param_pop_distr}. $\blacksquare$

%\vspace{-0.2in}
\section{Proof of \thrmref{thm:time_complexity_centralized_TL_convex_param}} \label{app:estimation_TL_convex_param}
\begin{figure*}[!t]
\begin{eqnarray}
\Pr\{\hat{\mathcal{T}}^* \geq \mathcal{T}^* + \epsilon\} &\leq& \Pr\left\{ \norm{\hat{\Theta}_\text{tl} - {\Theta}}_2 > \frac{\Omega}{C}\right\} \label{eq:tl_param_twoterms1} \\
&\stackrel{(a)}{\leq}&  \Pr\left\{\lambda \norm{\hat{\Theta}_\text{s} - {\Theta}}_2 + (1-\lambda) \norm{\hat{\Theta}_\text{t} - {\Theta}}_2 > \frac{\Omega}{C}\right\} \label{eq:tl_param_twoterms2} \\
\nonumber &=& \Pr\left\{\lambda \norm{\hat{\Theta}_\text{s} - {\Theta}}_2 + (1-\lambda) \norm{\hat{\Theta}_\text{t} - {\Theta}}_2 > \frac{\Omega}{C} \bigcap \mathcal{E}\right\} \\   && + \Pr\left\{\lambda \norm{\hat{\Theta}_\text{s} - {\Theta}}_2 + (1-\lambda) \norm{\hat{\Theta}_\text{t} - {\Theta}}_2 > \frac{\Omega}{C} \bigcap \mathcal{E}^c\right\} \label{eq:tl_param_twoterms3} \\
&\stackrel{(b)}{\leq}& \Pr\left\{\lambda \norm{\hat{\Theta}_\text{s} - {\Theta}}_2 > \frac{\Omega}{C} - D_t \right\} + \Pr\left\{  \norm{\hat{\Theta}_\text{t} - {\Theta}}_2 > \frac{D_t}{1-\lambda} \right\}, \label{eq:tl_param_twoterms4}
\end{eqnarray}
\hrulefill
\end{figure*}
We begin with
\begin{eqnarray*}
\Pr\{\hat{\mathcal{T}}^* \geq \mathcal{T}^* + \epsilon\}\!\!\! &\leq& \!\!\! \Pr \left\{ \sup_{0 \leq \pi \leq 1} g(\pi)\sum_{i=1}^N \abs{{p}_{{\hat{\Theta}_{\text{tl}},i}} - p_{\Theta,i}} > \tilde{\epsilon} \right\}\\ &\leq&  \Pr \left\{ \sum_{i=1}^N\abs{{p}_{{{\hat \Theta}_{\text{tl}}},i} - p_{\Theta,i}} > \Omega \right\},
\end{eqnarray*}
where ${p}_{{\hat \Theta}_{\text{tl}},i}$ is the estimate of ${p}_{{\Theta},i}$ using the TL-based approach described in \secref{subsec:param_tl}, $\Omega \triangleq  \frac{R_0 \epsilon}{2 B}$, and $g(\pi)\triangleq \exp\{-\lambda_u \pi \gamma^2 \left[1 - (1-\pi_i)^M \right]\}$. Note that, $\hat{\Theta}_\text{tl} \triangleq \lambda \hat{\Theta}_s + (1-\lambda) \hat{\Theta}_t$. Further, from the remainder form of the Taylor series around ${\Theta}$ (true parameter), we get $p_{\hat{\Theta}_{\text{tl},i}} = p_{{\Theta},i} + (\hat{\Theta}_\text{tl} - {\Theta}) \partial p_\Theta |_{\Theta \in [\Theta_\text{tl},\Theta]}$, $i =1,2,\ldots,N$, which implies that
\begin{eqnarray*}
\sum_{i=1}^N \abs{{p}_{{{\hat \Theta}_{\text{tl}}},i} - p_{\Theta,i}} &=& \sum_{i=1}^N \abs{(\hat{\Theta}_\text{tl} - {\Theta}) \partial p_\Theta |_{\Theta \in [\Theta_\text{tl},\Theta]}}\\ &\stackrel{(a)}{\leq}& \norm{\hat{\Theta}_\text{tl} - {\Theta}}_2 \sum_{i=1}^N\!\! \norm{\partial p_\Theta |_{\Theta \in [\Theta_\text{tl},\Theta]}}_2 \\ &\stackrel{(b)}{\leq}& C \norm{\hat{\Theta}_\text{tl} - {\Theta}}_2,
\end{eqnarray*}
where $(a)$ follows from Cauchy-Schwartz inequality and $(b)$ follows from \textbf{Assumption 1} in \secref{sec:param_pac_bound}. Therefore we have \eqref{eq:tl_param_twoterms1} - \eqref{eq:tl_param_twoterms4} at the top of this page, where $(a)$ follows from using $\hat{\Theta}_\text{tl} = \lambda \hat{\Theta}_s + (1-\lambda) \hat{\Theta}_t$ followed by the triangle inequality. Here, $\mathcal{E} \triangleq \{\norm{\hat{\Theta}_\text{t} - {\Theta}}_2 < \frac{D_t}{1-\lambda}\}$, and we let $D_t < \Omega/C$. The first term can be expressed as follows:
\begin{eqnarray}
&& \nonumber \Pr\left\{ \norm{\hat{\Theta}_\text{s} - {\Theta}}_2 > \frac{1}{\lambda}\left(\frac{\Omega}{C} - D_t \right) \right\}\\ \nonumber &\leq&  \Pr\left\{ \norm{\hat{\Theta}_\text{s} - {\Theta}_s}_2 + \norm{{\Theta}_\text{s} - {\Theta}}_2 > \bar{\Omega}^2 \right\} \\
\nonumber &\leq& \Pr\left\{ \sup_{1\leq i \leq d} \abs{\hat{\Theta}_\text{s,i} - {\Theta}_{s,i}}^2   > \left(\bar{\Omega} - \norm{{\Theta}_\text{s} - {\Theta}}_2\right)^2 \right\}\\
&\leq& d \Pr\left\{ \abs{\hat{\Theta}_\text{s,i} - {\Theta}_{s,i}}^2   > \left(\bar{\Omega} - \norm{{\Theta}_\text{s} - {\Theta}}_2\right)^2 \right\},
\label{eq:tl_param_twoterms5}
\end{eqnarray}
where $\bar{\Omega} \triangleq \frac{1}{\lambda}\left(\frac{\Omega}{C} - D_t\right) > \norm{{\Theta}_\text{s} - {\Theta}}_2$. However, $\hat{\Theta}_{s,i} \in [a,b]$ is an unbiased estimator of $\Theta_{s,i}$. Therefore, by Hoeffding's inequality, we can write
\begin{eqnarray}
\nonumber d \Pr\left\{ \abs{\hat{\Theta}_\text{s,i} - {\Theta}_{s,i}}^2   > \left(\bar{\Omega} - \norm{{\Theta}_\text{s} - {\Theta}}_2\right)^2 \right\} \\ \leq 2d \exp\left\{-\frac{2 m \left(\bar{\Omega} - \norm{{\Theta}_\text{s} - {\Theta}}_2\right)^2}{(b-a)^2} \right\}. \label{eq:tl_param_firstterm}
\end{eqnarray}

Next, we have
\begin{eqnarray}
\nonumber & \Pr\left\{  \norm{\hat{\Theta}_\text{t} - {\Theta}}_2 > \frac{D_t}{1-\lambda} \right\} \leq \\ & 2 d \exp\left\{-\lambda_u \pi R^2 \left(1 - \exp\{-\lambda_r \tau (1-e^{-\sigma_t^2})\}\right)\right\}, \label{eq:tl_param_secondterm}
\end{eqnarray}
where $\sigma_t^2:= \frac{D_t^2}{2(b-a)^2 (1-\lambda)^2}$, and the inequality follows from \eqref{eq:fsigma_ntl_param} by replacing $\Omega^2/dC^2$ with $\frac{D_t^2}{(1-\lambda)^2}$.

Therefore, $\Pr\{\hat{\mathcal{T}}^* \geq \mathcal{T}^* + \epsilon\}$ will be upperbounded by
\begin{eqnarray*}
&& 2d \exp\left\{-\frac{2 m \left(\bar{\Omega} - \norm{{\Theta}_\text{s} - {\Theta}}_2\right)^2}{(b-a)^2} \right\}\\ && + 2 d \exp\left\{-\lambda_u \pi R^2 \left(1 - e^{\{-\lambda_r \tau (1-e^{-\sigma_t^2})\}}\right)\right\},
\end{eqnarray*}
which is less than or equal to $\delta$ if

\begin{eqnarray*}
\tau \geq \frac{1}{\lambda_r (1-e^{-\sigma_t^2})}\log \left(\frac{1}{1 - \frac{1}{\lambda_u \pi R^2} \left(\log \frac{2 d}{\delta} + \log C_1 \right)}\right),
\end{eqnarray*}
for $\lambda_u > \frac{1}{\pi R^2}\left(\log \frac{2 d}{\delta} + \log \frac{1}{1 - \frac{2 d}{\delta} \exp\left\{-\frac{2m (\bar{\Omega} - \norm{\Theta - \Theta_s}_2)^2}{(b-a)^2}\right\}}\right)$, where
\begin{eqnarray*}
C_1 = \frac{1}{1 - \frac{2 d}{\delta} \exp\left\{-\frac{2m (\bar{\Omega} - \norm{\Theta - \Theta_s}_2)^2}{(b-a)^2}\right\}}
\end{eqnarray*}
These together with the conditions $\bar{\Omega} > \norm{{\Theta}_\text{s} - {\Theta}}_2$ and $\frac{2 d}{\delta} \exp\left\{-\frac{2m (\bar{\Omega} - \norm{\Theta - \Theta_s}_2)^2}{(b-a)^2}\right\} < 1$ proves \thrmref{thm:time_complexity_centralized_TL_convex_param}.  $\blacksquare$

\clearpage

\end{document}